%% file: paper.tex
\documentclass[a4paper,11pt]{llncs}
\usepackage[utf8]{inputenc}
\usepackage{amsmath}
\usepackage{amsfonts}
\usepackage{mathtools}
\usepackage{url}
\usepackage{enumerate}
\usepackage{multirow}
\usepackage{thmtools, thm-restate}
\usepackage[numbers,sectionbib]{natbib}

\DeclareMathOperator*{\argmin}{argmin}

\usepackage{array}
\newcolumntype{C}[1]{>{\centering\let\newline\\\arraybackslash\hspace{0pt}}m{#1}}
\urldef{\mailsa}\path|{vijay.menon, kate.larson}@uwaterloo.ca|

\begin{document}

\title{Reinstating Combinatorial Protections for Manipulation and Bribery in Single-Peaked and Nearly Single-Peaked Electorates}


 \author{Vijay Menon \and Kate Larson}
 \institute{David R. Cheriton School of Computer Science, University of Waterloo, \\ Waterloo, Ontario, Canada \\ \mailsa}


\maketitle

\input{abstract.tex} 
 \input{introduction.tex}

  \input{preliminaries.tex}

    \input{proofs.tex}

  \input{conclusion.tex}

 \bibliographystyle{splncsnat}
\bibliography{paper}

\end{document}

%% file: abstract.tex
\begin{abstract}
 Understanding when and how computational complexity can be used to protect elections against different manipulative actions has been a highly active research area over the past two decades. A 
recent body of work, however, has shown that many of the NP-hardness shields, previously obtained, vanish when the electorate has single-peaked or nearly single-peaked preferences. In light of 
these results, we investigate whether it is possible to reimpose NP-hardness shields for such electorates by allowing the voters to specify partial preferences instead of insisting they cast complete 
ballots. In particular, we show that in single-peaked and nearly single-peaked electorates, if voters are allowed to submit top-truncated ballots, then the complexity of manipulation and bribery for 
many voting rules increases from being in P to being NP-complete.
\end{abstract}

%% file: introduction.tex
\section{Introduction}
Collective decision making problems abound in human as well as multiagent contexts and they typically proceed by using a mechanism that aggregates the preferences of the participating 
agents. Voting is one such mechanism, and is, in fact, one of the most widely used ones. For instance, it has been proposed as a mechanism for web spam reduction \cite{dwork2001}, for collaborative 
filtering and recommender systems \cite{pennock2000}, and for multiagent planning \cite{ephrati1993}. As a result of its importance, voting has been extensively studied, and the theory of social 
choice has a number of impossibility results surrounding fundamental issues that arise in running elections. Among these, one aspect that has attracted considerable attention is the impact of 
different manipulative actions (bribery, control, and manipulation) on elections. Although the Gibbard-Satterthwaite theorem states that all reasonable voting rules are manipulable, starting with the 
seminal work of \citeauthor{bartholdi} \cite{bartholdi}, there has been much work that has looked into how computational complexity can be used as a barrier to protect elections against 
different manipulative actions (see \cite{faliszewski2010} for a survey).  

While there have been a lot of results in computational social choice that has obtained NP-hardness shields for different voting rules using constructions on combinatorially rich structures such as 
partitions and covers, a recent body of work, which was mainly inspired from the work of \citeauthor{walsh} \cite{walsh}, has shown that such combinatorial protections vanish when the 
voters have structured preferences. In particular, in single-peaked electorates it was observed by \citeauthor{faliszewski} \cite{faliszewski} for control and manipulation and by 
\citeauthor{brandt2010} \cite{brandt2010} for bribery that many of the previously known NP-hardness results fall into polynomial time. Subsequently, there has also been work done on the notion of 
nearly single-peaked preferences by \citeauthor{faliszewski2014} \cite{faliszewski2014} where similar, although not as stark, observations have been made. 


In the context of the above results, the present paper aims to take this line of research in a new direction by looking at the impact of partial preferences on manipulative actions in single-peaked 
and nearly single-peaked electorates. In particular, we consider top-truncated preferences, which are natural in settings where an agent is certain about his most preferred candidates, but is unsure 
or indifferent among the remaining ones, and we look at their impact on manipulative actions in single-peaked and nearly single-peaked settings. In doing so, we arrive at a number of surprising 
results, which in turn forms the theme of our paper -- of reinstating combinatorial protections by allowing top-truncated voting. Although, as noted by \citeauthor{faliszewski2014} 
\cite{faliszewski2014}, polynomial time algorithms for manipulative actions such as bribery and control are not always unethical or ``bad'' (as it can be a valuable tool in the hands of say the 
campaign manager, committee chair etc.), in majority of the situations they aren't socially ``good'' either and hence we believe that it is important to have elections protected against them. This 
thought, along with the ``easiness'' results previously obtained, and the fact that, among structured preference profiles, single-peaked and nearly single-peaked preferences are the most widely 
studied forms the main motivation of our paper. Here we look at how combinatorial protections that vanish in single-peaked and nearly single-peaked electorates can be reinstated in many cases by 
allowing top-truncated voting. 

Our contributions include the following:

\begin{enumerate}
 \item We, for the first time, systematically study the impact of partial voting on manipulative actions in structured preference profiles. In particular, we look at the problem of manipulation and 
bribery in single-peaked and nearly single-peaked settings when top-truncated ballots are allowed. 

 \item Under the assumption that the voters submit complete ballots, we first provide polynomial time algorithms for manipulation and weighted-bribery for certain voting rules in single-peaked and 
nearly single-peaked settings, thus extending the works of ~\citeauthor{faliszewski} \cite{faliszewski} for manipulation, \citeauthor{brandt2010} \cite{brandt2010} for bribery, and 
\citeauthor{faliszewski2014} \cite{faliszewski2014} for nearly single-peaked electorates. We then show how these polynomial-time problems become NP-complete when top-truncated ballots are 
allowed.

\item We show an example of a natural voting rule where, contrary to intuition, the complexity of manipulation actually increases when moving from the general case (i.e. when there is no restriction 
on the preferences) to the single-peaked case. In particular, in Theorem \ref{eVeto-SP} we show how the complexity of manipulating eliminate(veto), when top-truncated ballots are allowed, moves from 
being in $P$ 
in the general case to being NP-complete in the single-peaked case. 
\end{enumerate}

The overarching theme in this work is that top-truncated voting is useful in reinstating combinatorial protections in single-peaked and nearly single-peaked electorates. We believe that the results 
form a \emph{win-win} scenario: allowing voters to specify top-truncated ballots (or partial preferences, in general) is extremely useful and often necessary in many multi-agent systems applications 
and even in real-world elections, and allowing this additional flexibility in turn gives us what we want in terms of making the complexity of many manipulative-action problems hard.    

\subsubsection{Related Work} There are two lines of research that are closely related to our work. First is the work on structured preference profiles. This line of research has mainly looked at 
single-peaked preferences and more recently at nearly single-peaked preferences. The notion of single-peaked preferences was introduced by \citeauthor{black1948} \cite{black1948} and 
subsequently there has been a lot of work in social choice literature on the same. Among these, in particular, we note the work of \citeauthor{cantala2004} \cite{cantala2004} who introduced the 
concept of ``single-peaked with outside options" which is similar to the notion of single-peaked with top-truncated ballots that we study here, and the work of \citeauthor{barbera2007} 
\cite{barbera2007} who discussed how properties of different variants of single-peaked preferences change for varying amounts of indifference permitted.

In computational social choice, three papers that are most related to our work are \cite{faliszewski}, \cite{brandt2010}, and \cite{faliszewski2014}. The first two papers discuss manipulation and 
control, and bribery, respectively, and show how most of the NP-hardness shields for these manipulative actions vanish in single-peaked settings. The third paper studies the complexity of 
manipulative actions in nearly single-peaked electorates and shows how in many cases the hardness results evaporate. Our paper, in contrast, follows the theme of reinstating these combinatorial 
protections. 

Although the above mentioned papers are by far the most related to our work, it is worth noting that this line of research has mainly focused on the problem of single-peaked consistency where, 
informally, the task is to determine if a given set of preferences is single-peaked or otherwise (see \cite{bartholdi1986,doignon1994,escoffier2008,erdelyi2013,lackner2014}).  

The second line of research that is related to this paper is the work on election problems when partial preferences are allowed. Here, the papers that are most related are those of \citeauthor{nina} 
\cite{nina}, \citeauthor{fitzsimmons} \cite{fitzsimmons}, and \citeauthor{menon} \cite{menon}. \citeauthor{nina} \cite{nina} were the first to look at complexity of constructive 
manipulation under top-truncated voting and they provided an analysis for three particular voting protocols. Subsequently, \citeauthor{fitzsimmons} \cite{fitzsimmons} looked into how the 
complexity of bribery, control, and manipulation is affected when ties are allowed, and \citeauthor{menon} \cite{menon} generalized the complexity of constructive and destructive manipulation 
with top-truncated ballots for broader classes of voting rules and also looked at the impact on complexity when there is uncertainty about the non-manipulators' votes. While all three papers 
discuss results in the general setting (i.e. when there is no restriction on the structure of the preferences), in contrast to this, in this paper, we look at the complexity of manipulation and 
bribery with top-truncated ballots when the preferences are restricted to being single-peaked or nearly single-peaked. 

Additionally, we also note that there has been work on other election problems when preferences are only partially specified. Notable among them are the works of \citeauthor{konczak} 
\cite{konczak} and \citeauthor{xia} \cite{xia} on the possible and necessary winners problem, and the work of \citeauthor{baumeister} \cite{baumeister} which discusses planning 
various kinds of campaigns in settings where the ballots can be truncated at the top, bottom or both. The extension-bribery problem they introduce in that paper is closely related to the manipulation 
problem with top-truncated ballots.

%% file: preliminaries.tex
\section{Preliminaries}
\subsection{Elections, Voting Rules, and Preferences}
\subsubsection{Elections} An election is modeled as a pair $E = (C, V)$, where $C$ is the set of candidates and $V$ is the set of voter preferences. For every voter $v_i$, $\succ_{i}$ denotes their 
preference order over $C$. $\succ_{i}$ is said to be a complete order (or a complete vote) if it is antisymmetric, transitive, and a total ordering on $C$. Here, we also consider incomplete orders in 
the form of top-truncated ballots. $\succ_{i}$ is said to be a top-truncated order (or simply, a top order) if it is a complete order over any non-empty $C' \subseteq C$ and if all candidates in $C 
\setminus C'$ are assumed to be tied and are ranked below the candidates in $C'$. For simplicity, we sometimes use $(c_1, \cdots, c_m)$ to denote a preference order $c_1 \succ \cdots \succ c_m$. Since 
this paper looks at weighted elections, additionally, every voter $v_i$ has a weight $w_i$ associated with them.


\subsubsection{Voting Rules} An election system or a voting protocol takes the set of votes $V$ as input and it outputs a collection $W \subseteq C$ called the winner set. A candidate is said to 
be a Condorcet winner if it is preferred over every other candidate by a strict majority of the voters, while it is said be a weak-Condorcet winner if it is preferred over every other candidate by at 
least half of the voters. In this paper, we consider the following voting rules. We first present their original definitions which is on complete orders and then discuss how top orders can be handled.

\begin{enumerate}
 \item \textbf{Positional scoring rules:} A positional scoring rule is defined by a scoring vector $\alpha = \langle \alpha_1, \allowbreak \cdots, \alpha_m \rangle$, where $\alpha_1 \geq \cdots \geq 
\alpha_m$. For each voter $v$, a candidate receives $\alpha_i$ points if it is ranked in the $i$th position by $v$. 
Some examples of scoring rules are the \textit{plurality rule} with $\alpha = \langle 1, 0, \cdots, 0\rangle$, the \textit{Borda rule} with $\alpha = \langle m-1, m-2, \cdots, 0\rangle$, and the 
\textit{veto rule} with $\alpha = \langle 1, \cdots, 1, 0\rangle$.

\item \textbf{Scoring elimination rules:} Let $X$ be any scoring rule. Given a complete ordering, eliminate($X$) is the rule that successively eliminates the candidate with the lowest score according 
to $X$. Once a candidate is eliminated, the rule is then repeated with the reduced set of candidates until there is a single candidate left. In this paper, we mainly consider two scoring elimination 
rules: eliminate(Borda) -- which is also known as Baldwin's rule or Fishburn's version of Nanson's rule \cite{niou1987} -- and eliminate(veto). 


\item \textbf{Copeland$^\alpha$:} Let $\alpha \in \mathbb{Q}$, $0\leq \alpha \leq 1$. In Copeland$^{\alpha}$, introduced by \citeauthor{faliszewski2008} \cite{faliszewski2008}, for each pair of 
candidates, the candidate preferred by the majority receives one point and the other one receives a 0. In case of a tie, both receive $\alpha$ points. 
\end{enumerate}

A voting rule which, on every input that has a weak-Condorcet winner, outputs the set of all weak-Condorcet winners as the set of winners is said to be weak-Condorcet consistent. 

 To deal with top-truncated orders in positional scoring rules where a voter ranks only $k$ out of the $m$ candidates ($k < m$), we use the following three schemes that were used by 
\citeauthor{nina} \cite{nina} in their preliminary work on manipulation with top orders. \citeauthor{emerson} \cite{emerson} also used the same schemes for the Borda 
rule. 

\begin{enumerate}
 \item \textbf{Round-up:} A candidate ranked in the $i$th position ($i\leq k)$ receives a score of $\alpha_i$, while all the unranked candidates receive a score of $\alpha_m$. For any positional 
scoring rule $X$, we denote this by $X_{\uparrow}$.  

 \item \textbf{Round-down:} A candidate ranked in the $i$th position ($i\leq k)$ receives a score of $\alpha_{m-(k-i)-1}$, while all the unranked candidates receive a score of $\alpha_m$. For any 
positional scoring rule $X$, we denote this by $X_{\downarrow}$. 

 \item \textbf{Average score:} A candidate ranked in the $i$th position ($i\leq k)$ receives a score of $\alpha_{i}$, while all the unranked candidates receive a score of $\frac{\textstyle\sum_{k < j 
\leq m}\alpha_j}{m - k}$. For any positional scoring rule $X$, we denote this by $X_{av}$. 
\end{enumerate}


In scoring elimination rules, we deal with top-truncated votes by using the round-up scheme described above. Here, we consider a vote to be valid only until at least one of the candidates listed in it 
is remaining in the election. In other words, we simply ignore a vote once all the candidates listed in it are eliminated. In the case of Copeland$^\alpha$, top orders are dealt by just keeping to 
the definition which assumes that all the unranked candidates are tied and are ranked below the ranked candidates.

\subsubsection{Single-Peaked Preferences} The notion of single-peaked preferences, first introduced by \citeauthor{black1948} \cite{black1948}, captures settings where the preferences of a voter are 
based on a one-dimensional axis. The basic idea here is that every voter has a peak (their most preferred alternative) and that their utility for an alternative decreases the further it is away from 
this peak. Below, we provide the formal definition of single-peaked preferences on complete orders and then discuss what single-peakedness means when we allow top-truncated voting. We use the 
definition of single-peakedness found in \cite{faliszewski}. 

\begin{definition}
 A collection of votes, $V$, is said to be single-peaked if there exists a linear order $L$ over $C$ such that for every triple of candidates $a, b,$ and $c$ it holds that:
 \begin{equation*}
(a L b L c \vee c L b L a) \implies (\forall v \in V)\left[a \succ_{v} b \implies b \succ_{v} c \right].
 \end{equation*}    
\end{definition}

When voters are allowed to present top-truncated ballots, this notion of single-peakedness essentially captures those scenarios where they have a continuous range over $L$ over which their 
preferences are single-peaked and outside of which they are indifferent among the alternatives. In social choice theory, this notion of single-peaked preferences has been captured as single-peaked 
with outside options in the context of choosing a level of public good by \citeauthor{cantala2004} \cite{cantala2004}.

Throughout this paper, following the model proposed by \citeauthor{walsh} \cite{walsh}, we assume that the societal order $L$ is given as part 
of the input.  

\begin{example}
 Let $C = \{c_1, c_2, c_3, c_4\}$ and $c_2 L c_4 L c_3 L c_1$ be a linear order over $C$. Then, the preference orders $c_4 \succ c_3 \succ c_2 \succ c_1$ and $c_4 \succ c_2 \succ c_3 \succ c_1$ are 
both valid complete single-peaked orders, while the preference order $c_4 \succ c_1 \succ c_2 \succ c_3$ is not a valid single-peaked order. Also, with respect to the given linear order, $c_4 \succ 
c_3$ and $c_4 \succ c_2$ are both valid top-truncated single-peaked orders, while $c_4 \succ c_1$ is not. 
\end{example}

\subsubsection{Nearly Single-Peaked Preferences} Although single-peaked preferences are an interesting domain to study, it is often the case that real-world electorates are not 
truly single-peaked, but are only very close to being single-peaked. The notion of ``near'' single-peakedness was first raised by \citeauthor{conitzer2007} \cite{conitzer2007} and 
\citeauthor{escoffier2008} \cite{escoffier2008}, and was subsequently systematically studied by \citeauthor{faliszewski2014} \cite{faliszewski2014} and \citeauthor{erdelyi2013} 
\cite{erdelyi2013}. In this paper, we look at only one notion of ``nearness'', namely the maverick notion which is defined below.

\begin{definition}[$k$-maverick SP Electorate]
A collection of votes $V$ is called a $k$-maverick SP electorate if all but at most $k$ of the voters are single-peaked consistent with the societal order $L$. 
\end{definition}

\subsection{Manipulative Actions}
In this paper, we consider two manipulative action problems: manipulation and bribery. In particular, we study the Constructive Coalitional Weighted Manipulation (CWCM) problem and the 
Weighted-bribery problem. CWCM was first studied by \citeauthor{conitzer} \cite{conitzer} and is described below. 

\begin{definition}[CWCM]
Given a set of candidates, $C$, a set of weighted votes, $S$ (preferences of the non-manipulators), the weights for a set of votes, $T$ (manipulators' votes), and a preferred candidate, $p$, we 
are asked if there exists a way to cast the votes in $T$ so that $p$ is a winner in the election $E = (C, S \cup T)$. 
\end{definition}

The complexity-theoretic study of the bribery problem was first introduced by \citeauthor{faliszewski2009} \cite{faliszewski2009}. 
Here we mainly look at the weighted version of the bribery problem which is described below. 

\begin{definition}[Weighted-bribery]
 Given a set of candidates, $C$, the set of weighted votes, $V$, a preferred candidate, $p$, and a limit, $k \in \mathbb{N}$, we are asked if there exists a way to change the votes 
of at most $k$ of the voters in $V$ so that it results in $p$ being a winner.   
\end{definition}

Throughout this paper, unless otherwise specified, we use the non-unique winner model (where the objective is to make the preferred candidate \textit{a} winner) as our standard model.

\subsection{Computational Complexity}
In most of the proofs for NP-hardness in this paper, we use reductions from either the well-known NP-complete problem Partition, from a variant of the Partition problem (Partition'), or from a 
variant of the subset sum problem (Fixed-Difference Subset Sum). 

\begin{definition}[{Partition}]
 Given a set of non-negative integers $S = \{a_i\}_{1 \leq i \leq n}$ summing to $2K$, we are asked if there exists a subset $S_1$ of $S$ which sums to $K$. 
\end{definition}

\begin{definition}[Partition']
 Given a set of non-negative integers $S = \{a_i\}_{1 \leq i \leq n}$ such that $\sum_i a_i = 2nK$ and $a_i \geq K$, we are asked if there exists a subset $S_1$ such that $\sum S_1 = nK$. 
\end{definition}

The NP-completeness of Partition' can be shown by reduction from the Partition problem.


\begin{restatable}{lemma}{parti}
 Partition' is NP-complete.
\end{restatable}

\begin{definition}[{Fixed-Difference Subset Sum}]
 Given a set of non-negative integers $S = \{a_i\}_{1 \leq i \leq n}$ summing to $2K$, we are asked if there exists two disjoint subsets $S_1$, $S_2$ of $S$ such that $\sum S_1 - \sum S_2 = K$, where 
$\sum S_i$ denotes the sum of all the elements in the set $S_i$. 
\end{definition}

The proof of NP-completeness of Fixed-Difference Subset Sum can be found in \cite{menon}.

%% file: proofs.tex
\newif\ifdraft
 \drafttrue

\section{Manipulation} \label{manipulation}
In this section we study CWCM with top-truncated votes in both single-peaked and nearly single-peaked electorates. Since the theme of this paper is the reinstatement of combinatorial protections by 
top-truncated voting, for all the voting rules considered in this section, we present both the ``easiness'' result (if not already known from previous work) as well as the subsequent ``hardness'' 
result that arises as a consequence of allowing top-truncated ballots. 

\subsection{Single-Peaked Electorates} \label{sec-SP}
\citeauthor{walsh} \cite{walsh} was the first  to consider manipulation with single-peaked preferences and he showed that STV remains NP-hard to manipulate for 3 candidates. Subsequently, 
\citeauthor{faliszewski} \cite{faliszewski} showed that for many voting protocols which are usually hard to manipulate, restricting the preferences to being single-peaked makes them easy. In 
particular, they showed that any 3-candidate scoring rule with $(\alpha_1 - \alpha_3) \leq 2(\alpha_2 - \alpha_3)$ is easy to manipulate. This result was then extended to obtain a complete 
characterization for any $m$-candidate scoring rule in \cite{brandt2010}. Here, we look at 3-candidate scoring rules again and we study the impact on complexity of manipulation when top-truncated 
voting is allowed.



\begin{restatable}{theorem}{spSCsd} \label{SP-SC-D}
 For any 3-candidate scoring rule $X$ that is not isomorphic to plurality or veto, in single-peaked electorates, CWCM with top-truncated votes in $X_{\downarrow}$ is NP-complete.
\end{restatable}

\begin{proof}
Since there are only three candidates, the scoring vector for the corresponding positional scoring rule is defined by $\langle \alpha_1, \alpha_2, \alpha_3 \rangle$, where $\alpha_1 > \alpha_2 > 
\alpha_3 = 0$ (because $\alpha_1 = \alpha_2$ is isomorphic to veto, $\alpha_2 = \alpha_3$ is isomorphic to plurality, and $\alpha_3$ can be taken to be zero since translating the scores in a 
scoring rule does not affect the outcome of the rule).

 It is easy to see that the problem is in NP. To prove NP-hardness, we proceed by considering the following three cases. In all the cases we use a reduction from the Partition problem.
 
 \item{\textbf{Case 1: $\alpha_2 < \alpha_1 < 2\alpha_2$.}} Consider the following instance of CWCM, where $a, b,$ and $p$ are the three candidates. Let $p L a L b$ be the linear ordering over the 
candidates. This linear ordering in turn restricts the set of allowed votes to $\{(p, a, b), (a, p, b), (a, b, p), (b, a, p), (a), (b), (p)\}$. In $S$, let there be a vote of weight $K$ voting for 
$(a, b, p)$, $(b, a, p)$, $(b)$, and $(p)$ each. As a result the scores of $a, b,$ and $p$ are $\alpha_1K  + \alpha_2K$, $\alpha_1K  + 2\alpha_2K$, and $\alpha_2K$ respectively. In $T$, let each 
$k_i$ 
have a vote of weight $k_i$. 

Suppose there exists a partition. Let those manipulators in the first half vote $(p, a, b)$ and let the others vote $(p)$. Then the scores of $a, b,$ and $p$ are $\alpha_1K  + 2\alpha_2K$, 
$\alpha_1K  + 2\alpha_2K$, and $\alpha_1K + 2\alpha_2K$ respectively, and hence $p$ wins by tie-breaking. 

Conversely, suppose there exists a manipulation in favor of $p$. Let $x$ and $y$ denote the total weight of the manipulators voting for $(p, a, b)$ and $(p)$, respectively. Since there exists a 
successful manipulation in favor of $p$, the score of $p$ should be at least as much as that of $a$. Therefore, we have: $x\alpha_1 + y\alpha_2 + \alpha_2K \geq \alpha_1K  + x\alpha_2 +
\alpha_2K$. Using the fact that $x + y = 2K$, this simplifies to $(x - K)(\alpha_1 - 2\alpha_2) \geq 0$. But since we assumed that $\alpha_1 < 2\alpha_2$, we have $x \leq K$. Similarly, the score of 
$p$ should also be at least as much as $b$ and so we have $x\alpha_1 + y\alpha_2 + \alpha_2K \geq \alpha_1K  + 2\alpha_2K$. Simplifying this we have $(x - K)(\alpha_1 - \alpha_2) \geq 0$, and since 
we assumed that $\alpha_2 < \alpha_1$, therefore $x \geq K$. But then, $x$ cannot be both greater than and lesser than $K$ at the same time. So $x$ has to equal to $K$ and this in turn implies that 
$y = K$ and that there exists a partition. 

 \item{\textbf{Case 2: $\alpha_1 > 2\alpha_2$.}} Consider the following instance of CWCM, where $a, b,$ and $p$ are the three candidates. Let $a L p L b$ be the linear ordering over the 
candidates. This linear ordering in turn restricts the set of allowed votes to $\{(a, p, b), (p, a, b), \allowbreak (p, b, a), (b, p, a), (a), (b), (p)\}$. In $S$, let there be a vote of weight 
$(2\alpha_1 - \alpha_2)K$ voting for $(a, p, b)$, and $(b, p, a)$ each. As a result the scores of $a, b,$ and $p$ are $(2\alpha_1 - \alpha_2)\alpha_1K$, $(2\alpha_1 - \alpha_2)\alpha_1K$, and 
$2(2\alpha_1 - \alpha_2)\alpha_2K$ respectively. In $T$, let each $k_i$ have a vote of weight $(\alpha_1 - 2\alpha_2)k_i$. 


Suppose there exists a partition. Let those manipulators in the first half vote $(p, a, b)$ and let the others vote $(p, b, a)$. Then the scores of $a, b,$ and $p$ are all $2(\alpha_1^{2} - 
\alpha_2^{2})K$, and hence $p$ wins by tie-breaking. 

Conversely, suppose there exists a manipulation in favor of $p$. Let $x, y$, and $z$ be the sum of the $k_i$'s of the manipulators in $T$ who vote $(p, a, b), (p, b, a),$ and $(p)$, respectively. 
Since there exists a successful manipulation in favor of $p$, the score of $p$ should be at least as much as that of $a$. Therefore, we have: $2(2\alpha_1 - \alpha_2)\alpha_2K + 
(x+y)\alpha_1(\alpha_1 
- 2\alpha_2) + z\alpha_2(\alpha_1 - 2\alpha_2) \geq (2\alpha_1 - \alpha_2)\alpha_1K + \alpha_1K  + x\alpha_2(\alpha_1 - 2\alpha_2)$. Using the fact that $x + y + z = 2K$, this simplifies to 
$\alpha_2K - z(\alpha_1 - \alpha_2) \geq x\alpha_2$ (1). Doing the same with respect to $p$ and $b$ we have, $\alpha_2K - z(\alpha_1 - \alpha_2) \geq y\alpha_2$ (2). Adding (1) and (2) we have 
$z(2\alpha_1 - 3\alpha_2) \leq 0$. Since we assumed that $\alpha_1 > 2\alpha_2$, $(2\alpha_1 - 3\alpha_2) > 0$ and this in turn implies that $z \leq 0$. But then $z$ cannot be less than zero and so 
it should be zero. Substituting $z=0$ in (1) and (2) we have $x \leq K$, and $y \leq K$ respectively. Since $x + y + z = 2K$, this implies that $x = K$ and $y = K$ and that there exists a partition.

\item{\textbf{Case 3: $\alpha_1 = 2\alpha_2$.}} Consider the following instance of CWCM, where $a, b,$ and $p$ are the three candidates. Let $a L p L b$ be the linear ordering over the 
candidates. This linear ordering in turn restricts the set of allowed votes to $\{(a, p, b), (p, a, b), \allowbreak (p, b, a), (b, p, a), (a), (b), (p)\}$. In $S$, let there be a vote of weight 
$3K$ voting for $(a)$, and $(b)$ each. As a result the scores of $a, b,$ and $p$ are $3\alpha_2K$, $3\alpha_2K$, and 0 respectively. In $T$, let each $k_i$ have a vote of weight $k_i$. 

Suppose there exists a partition. Let those manipulators in the first half vote $(p, a, b)$ and let the others vote $(p, b, a)$. Then the scores of $a, b,$ and $p$ are $4\alpha_2K$ 
$4\alpha_2K$, and $2\alpha_1K$, respectively. Since $\alpha_1 = 2\alpha_2$, all of them have the same score and hence $p$ wins by tie-breaking. 

Conversely, suppose there exists a manipulation in favor of $p$. Let $x, y$, and $z$ be total weight of the manipulators in $T$ who vote $(p, a, b), (p, b, a),$ and $(p)$, respectively. Since there 
exists a successful manipulation in favor of $p$, the score of $p$ should be at least as much as that of $a$. Therefore, we have: $(x+y)\alpha_1 + z\alpha_2 \geq 3\alpha_2K + x\alpha_2$. Using the 
fact that $x + y + z = 2K$, and that $\alpha_1 = 2\alpha_2$ this simplifies to $y \geq K$. Doing the same with respect to $p$ and $b$ we have, $x \geq K$. But then since $x + y + z = 2K$, this 
implies that $x = K$ and $y = K$ and that there exists a partition.  \qed
\end{proof}



\begin{restatable}{theorem}{spSCav} \label{SP-SC-AV}
 For any 3-candidate scoring rule $X$ that is not isomorphic to plurality, in single-peaked electorates, CWCM with top-truncated votes in $X_{av}$ is NP-complete.
\end{restatable}

\begin{proof}
 Clearly the problem is in NP. To show NP-hardness, we reduce from an arbitrary instance of the Fixed-difference Subset Sum problem. 
 
 Consider the following instance of CWCM, where $a, b,$ and $p$ are the three candidates. Let $a L p L b$ be the linear ordering over the candidates. This linear ordering in turn restricts the set of 
allowed votes to $\{(a, p, b), (p, a, b), (p, b, a), \allowbreak (b, p, a), (a), (b), (p)\}$. In $S$, let there be a vote of weight $8K$ voting for $(a)$, and $(b)$ each, vote of weight $3K$ voting 
for $(p, b, a)$, and a vote of weight $K$ voting $(p, a, b)$. As a result the scores of $a, b,$ and $p$ are $8\alpha_1K + 5\alpha_2K$, $8\alpha_1K + 7\alpha_2K$, and $4\alpha_1K + 8\alpha_2K$ 
respectively. In $T$, let each $k_i$ have a vote of weight $2k_i$. 

Suppose there exists $S_1$, $S_2$ such that $\sum S_1 - \sum S_2 = K$. Let those manipulators who are in $S_1$ vote $(p, a, b)$, those in $S_2$ vote $(p, b, a)$, and let all the remaining 
manipulators 
vote $(p)$. If $x, y,$ and $z$ denote the sum of $k_i$'s of the manipulators who vote for $(p, a, b), (p, b, a),$ and $(p)$, respectively, then the scores of $a, b,$ and $p$ are $8\alpha_1K + 
5\alpha_2K + 2x\alpha_2 + z\alpha_2$, $8\alpha_1K + 7\alpha_2K + z\alpha_2$, and $8\alpha_1K + 8\alpha_2K$, respectively. Now, since $x - y = K$ and $x + y + z = 2K$, we see that scores of $a$, $b$, 
and $p$ are the same and hence $p$ wins by tie-breaking. 

Conversely, suppose there existed a successful manipulation for $p$. From above we know that  if $x, y,$ and $z$ denote the sum of $k_i$'s of the manipulators who vote for $(p, a, b), (p, b, a),$ and 
$(p)$, respectively, then the scores of $a, b,$ and $p$ are $8\alpha_1K + 5\alpha_2K + 2x\alpha_2 + z\alpha_2$, $8\alpha_1K + 7\alpha_2K + z\alpha_2$, and $8\alpha_1K + 8\alpha_2K$, respectively. 
Since there is a successful manipulation the score of $p$ should be at least as much as that of $a$, and so this in turn implies that $x - y \leq K$. Similarly, doing it with respect to $p$ and $b$ 
we have $x - y \geq K$. But then $x - y$ cannot be both greater and smaller than $K$ at the same time. So this implies that $x - y = K$ and that there exists two subsets such that $\sum S_1 - \sum 
S_2 = K$, where $x = \sum S_1$ and $y = \sum S_2$. \qed
\end{proof}

From Theorem \ref{SP-SC-D} and Theorem \ref{SP-SC-AV} we can see that a relaxation of the complete votes assumption by additionally allowing top-truncated votes actually increases the complexity of 
CWCM for all 3-candidate scoring rules with $(\alpha_1 - \alpha_3) \leq 2(\alpha_2 - \alpha_3)$ from being in P \cite{faliszewski} to being NP-complete when either the round-down or 
average score evaluation schemes are used. However, with the round-up evaluation scheme manipulation become easy for all $m$-candidate scoring rules as shown by the theorem below. 

\begin{restatable}{theorem}{spOthers} \label{SP-SC-U}
 In single-peaked electorates, computing if a coalition of manipulators can manipulate plurality$_{\downarrow}$, veto$_{\downarrow}$, plurality$_{av}$, and $X_{\uparrow}$, for any scoring rule $X$,  
with weighted top-truncated votes takes polynomial time (for any number of candidates).
\end{restatable}

\begin{proof}
 For all the voting rules except plurality$_{\downarrow}$, the manipulators can simply check if all of them voting $(p)$ will make $p$ a winner. If not they cannot make $p$ a winner. In case of 
plurality$_{\downarrow}$, all the manipulators can vote for any single-peaked order consistent with the societal order $L$ that has $p$ at the top. \qed
\end{proof}

Another interesting point to note here is that Theorem \ref{SP-SC-D}, Theorem \ref{SP-SC-AV}, and Theorem \ref{SP-SC-U} together also imply that the restriction of preferences to being single-peaked 
has no effect on the complexity of manipulation with top-truncated ballots, since the same results were obtained by \citeauthor{menon} \cite{menon} in the general case as well. 

Next, we look at CWCM in Copeland$^{\alpha}$ and we present both the ``easiness'' and the ``hardness'' result. 

\begin{restatable}{theorem}{copeland}
 In single-peaked electorates, for 3-candidate Copeland$^{\alpha}$, $\alpha \in \mathbb{Q}$, $0 \leq \alpha < 1$,
 \begin{enumerate}
  \item CWCM with \textbf{complete votes} is in $P$.
  \item CWCM with \textbf{top-truncated votes} is NP-complete.
 \end{enumerate}
\end{restatable}

\begin{proof}
 We first prove the polynomial time result for CWCM with complete votes and then show the NP-completeness of CWCM with top-truncated votes.
 Let $a, b$, and $p$, be the three candidates. Without loss of generality, we need to consider only the following two linear orderings $L$:
 \begin{enumerate}[i)]
  \item $p L a L b$: In this case the only strategy for all the manipulators is to just vote $(p, a, b)$.
  \item $a L p L b$: In this case the linear ordering restricts the set of allowed votes to $\{(a, p, b), \allowbreak (p, a, b), (p, b, a), (b, p, a)\}$. As a result, it can be seen that regardless 
of the weights of the voters in $S$, $p$ cannot lose to both $a$ and $b$ in a pairwise election. Therefore, the only possibilities are $p$ losing to $a$ (respectively, $b$), but winning against $b$ 
($a$), or $p$ tying against both $a$ and $b$. In the latter scenario all the manipulators can either vote $(p, a, b)$ or $(p, b, a)$, while in the former all of them can vote $(p, b, a)$ 
(respectively, $(p, a, b)$). 
\end{enumerate}  
 
  It is easy to see that the problem is in NP. To be prove NP-hardness, we reduce from an arbitrary instance of the Fixed-difference Subset Sum problem.
  
   Consider the following instance of CWCM, where $a, b,$ and $p$ are the three candidates. Let $a L p L b$ be the linear ordering over the candidates. This linear ordering in turn restricts the set 
of allowed votes to $\{(a, p, b), (p, a, b), (p, b, a), \allowbreak (b, p, a), (a), (b), (p)\}$. In $S$, let there be a vote of weight $6K$ voting for $(a)$ and $(b)$ each, and a vote of weight $2K$ 
voting for $(p, a, b)$. In $T$, let each $k_i$ have a vote of weight $2k_i$.

Suppose there exists $S_1, S_2$ such that $\sum S_1 - \sum S_2 = K$. In Copeland$^{\alpha}$, it can be assumed that all the manipulators rank $p$ first. So, let the manipulators in $S_1$ vote $(p, 
b, a)$, those in $S_2$ vote $(p, a, b)$, and let the rest vote $(p)$. If $N_{V}(r, s)$ denotes the total number of votes in $V$ which rank $r$ prior to $s$, and $D_{V}(r, s) = N_{V}(r, s) - 
N_{V}(s, r)$, then $D_{S \cup T}(p, a) = 0$ and $D_{S \cup T}(p, b) = 0$. Therefore, the score of $p$, $s(p) = 2\alpha$. Also since $\sum S_1 - \sum S_2 = K$,  $D_{T}(a, b) = -2K$, while $D_S(a,b) = 
2K$. Therefore, $D_{S \cup T}(a, b) = 0$ and so, both receive a score $2\alpha$. As a result, $p$ wins by tie-breaking.

Conversely, suppose there exists a successful manipulation in favor of $p$. If $x, y$, and $z$, denote the sum of $k_i$'s of the manipulators in $T$ who vote $(p, a, b), \allowbreak (p, b, a),$ and 
$(p)$, respectively, then without taking into account the pairwise election between $a$ and $b$ in $T$, the score of $p$, $a$, and $b$ is $2\alpha, 1 + \alpha,$ and $\alpha$, respectively. Now since 
$2\alpha < 1 + \alpha$ for all rational $\alpha \in [0, 1)$, therefore, the only way $p$ would win this is if including the pairwise election between $a$ and $b$ in $T$ results in a tie between them. 
So this implies that $D_{S \cup T}(a, b) = 2K + 2x - 2y = 0$, or $y - x = K$. \qed
\end{proof}

 
We note that for $\alpha = 1$ both CWCM with complete votes and top-truncated votes can be shown to be in $P$.

Our final result for manipulation under single-peaked preferences is the very interesting case of eliminate(veto). 

\begin{restatable}{theorem}{eVetoSP}
\label{eVeto-SP}
  In single-peaked electorates, in the unique winner model, for eliminate(veto), 
  \begin{enumerate}
   \item CWCM with \textbf{complete votes} is in $P$ when the number of candidates is bounded.
   \item CWCM with \textbf{top-truncated votes} is NP-complete for even three candidates.
  \end{enumerate}
\end{restatable}

\begin{proof}
We first present the polynomial time algorithm for CWCM with complete votes.

 Suppose there was an arbitrary set $W$ of the manipulators' votes which -- along with $S$ -- resulted in $p$ winning. Let the corresponding elimination order be $e = (c_1, \cdots, c_m = p)$. To 
prove this theorem, we need the following lemmas. 

\begin{lemma} \label{l-eV1}
 At each round of the eliminate(veto) only the rightmost or the leftmost candidate in the linear ordering $L$ is eliminated. 
\end{lemma}

\begin{proof}
Let $c_1 L c_2 L \cdots L c_m$ be any arbitrary linear ordering. It is clear that in the first round of eliminate(veto) only $c_1$ or $c_m$ will be eliminated since all the votes only have either of 
them placed at the end. Also, since the votes are single-peaked, for any candidate $c_j$, the only candidates which can be immediately on top of it in a vote where $c_j$ is placed last are its right 
and left neighbors, where the right neighbor of the rightmost candidate is the leftmost one and vice-versa. Therefore, in particular, if $c_1$ is the candidate eliminated in the first round then all 
its last votes gets transferred to either $c_2$ or $c_m$. Hence in the subsequent round only one of $c_2$ or $c_m$ will be eliminated. Continuing this way we can see that at every round $r$ only one 
of the corner candidates in $L$ (which now is the linear ordering over the remaining $m -r + 1$ candidates) will be eliminated.  \qed
\end{proof}

\begin{lemma} \label{l-eV2}
 The reverse of $e$ is a single-peaked order with respect to the given linear order $L$.
\end{lemma}
\begin{proof}
Our claim is that $O: c_m = p \succ c_{m-1} \succ \cdots \succ c_1$ is a valid single-peaked order with respect to $L$. To prove this, let us assume the contrary. This implies that there exists three 
candidates $c_i$, $c_j$, $c_k$ such that $c_j \succ_{O} c_i \succ_{O} c_k$, but according to the linear ordering $L$ it is only possible that $c_i L \cdots L c_k L \cdots L c_j$ or $c_j L \cdots L 
c_k L \cdots L c_i$. Now, $c_j \succ_{O} c_i \succ_{O} c_k$ implies that, among $c_i, c_j,$ and $c_k$, $c_k$ is eliminated first, followed by $c_i$, and then $c_j$. But then, since at any round $r$ 
only the corner elements can be eliminated (Lemma \ref{l-eV1}), $c_k$ can never be eliminated before both $c_i$ and $c_j$ in either of the two linear orderings. Hence $e$ cannot be a valid 
elimination 
order if it is not single-peaked. 
\qed
\end{proof}

\textbf{Proof of Theorem \ref{eVeto-SP} contd.} Now since the reverse of every elimination order $e$ caused by any arbitrary manipulation in favor of $p$ is single-peaked (Lemma \ref{l-eV2}), the 
manipulators can try out all possible single-peaked orders which have $p$ at the top and can try to induce $e$ by collectively voting for the same. This is enough because we know that in 
eliminate(veto) any elimination order that is achieved by an arbitrary set of manipulators' votes can be achieved if all the manipulators vote the same way (by using the reverse of the elimination 
order) \cite[Lemma 12]{coleman}. And since, when the number of candidates are bounded, there are only polynomial number of single-peaked orders which have $p$ at the top, this can be done in 
polynomial time.

We now go on to show the hardness of CWCM with top-truncated votes for 3-candidate eliminate(veto). Clearly, the problem is in NP. To show NP-hardness, we reduce from an arbitrary instance of 
the Partition problem. 
 
  Consider the following instance of CWCM, where $a, b,$ and $p$ are the three candidates. Let $a L b L p$ be the linear ordering over the candidates. This linear ordering in turn restricts the set 
of allowed votes to $\{(a, b, p), (b, a, p), (b, p, a), \allowbreak (p, b, a), (a), (b), (p)\}$. In $S$, let there be a vote of weight $K + 2$ voting for $(b, p, a)$, vote of weight $K - 1$ voting 
for $(b, a, p)$, a vote of weight $1$ voting for $(a, b, p)$, a vote of weight 2 voting $(p)$, and a vote of weight 3 voting $(a)$. In $T$, let each $k_i$ have a vote of weight $k_i$.

The first thing to observe is, $p$ cannot be unique winner if $a$ is eliminated in the first round. This is because once $a$ is eliminated, $score(b) = score(p) = 2K + 2$. Therefore, for $p$ to be a 
unique winner, $b$ has to be eliminated in the first round. This in turn implies that the total score of $a$ and $p$ should be at least 1 more than that of $b$ (since we are in the unique winner 
model). So we have, 
\begin{align*}
 &score_{T}(a) + K + 3 \geq 2K + 3 + score_T(b)\implies score_{T}(a) - score_{T}(b) \geq K \\
 &score_{T}(p) + K + 4 \geq 2K + 3 + score_T(b)\implies score_{T}(p) - score_{T}(b) \geq K-1.
\end{align*}
where $score_{T}(p)$ is the score $p$ receives from the votes in $T$. 
 
 Now, it is possible to show that the above equations can be satisfied only if $b$ receives a score of 0 from $T$. Therefore, this implies that all the votes are concentrated on $(p)$ and 
$(a)$. As a result, given the equations above, we basically have two possibilities: a total weight of $K+1$ voting for $(a)$, $K-1$ voting for $(p)$, or a total weight of $K$ voting for $(a)$ and 
$(p)$ each. In the first case it is easy to see that $p$ will be eliminated in the second round, while in the latter case $p$ will be a unique winner. Hence $p$ can be a unique winner if and only if 
there exists a partition. \qed
\end{proof}

Theorem \ref{eVeto-SP} is most interesting not because of the fact that it follows the theme of our paper, but for the following other reasons. First is the very unusual behavior that it is showing 
here. Eliminate(veto), when there are only a bounded number of candidates and in the unique winner model, is known to be in $P$ for practically everything -- from CWCM with complete votes in the 
general case \cite{coleman}, to CWCM with top-truncated votes in the general case \cite{menon}, and to even when there is only partial information (in the form of top-truncated votes) on the 
non-manipulators' votes \cite{menon}. However, here, with single-peaked preferences and with top-truncated votes, it is NP-complete even when there are only three candidates. Second, what makes 
Theorem \ref{eVeto-SP} even more interesting is the fact that this actually serves as a counterexample (with the only caveat that, to be fair, they had considered only complete votes in that paper) 
to 
a conjecture stated by \citeauthor{faliszewski} \cite{faliszewski} where they say that they do not expect the complexity of manipulation for ``any existing, natural voting system'' to increase 
when moving from the general case (where there is no restriction on the preferences) to the single-peaked case. But this is exactly what we are seeing here. 

\subsection{Nearly Single-Peaked Preferences}
For nearly single-peaked electorates, \citeauthor{faliszewski2014} \cite{faliszewski2014} were the first to look at the complexity of manipulation, bribery, and control. In that paper, they 
introduced several notions of nearness and among them was the $k$-maverick-SP-society where all but at most $k$ of the voters are consistent with the societal order $L$. As noted before, we only 
consider this notion of ``nearness'' in this paper. We start off by looking at 3-candidate scoring rules and we show the impact of top-truncated voting on CWCM. Note that Faliszewski et al. showed 
that for all 3-candidate scoring rules that are not isomorphic to plurality CWCM for 1-maverick-SP-societies was NP-complete \cite{faliszewski2014}. The NP-completeness proofs in this section are all 
based on the corresponding results for the case of single-peaked electorates in Section \ref{sec-SP}. 


\begin{restatable}{theorem}{nspSCd}
 \label{NSP-SC-D}
  In 1-maverick-SP societies, for any 3-candidate scoring rule $X$ that is not isomorphic to plurality or veto, CWCM with top-truncated votes in $X_{\downarrow}$ is NP-complete.
\end{restatable}

\begin{proof}
 Proceed the same way as in Theorem \ref{NSP-SC-D} by considering three cases.
 
  \item{\textbf{Case 1: $\alpha_2 < \alpha_1 < 2\alpha_2$.}} Construct the following instance of CWCM with $pLaLb$ as the linear ordering and with $S$ having a vote of weight $K$ voting for $(a, b, 
p)$, $(b, p, a)$, $(a)$, and $(b)$ each. Note that here $(b, p, a)$ is the maverick. In $T$, let each $k_i$ have a vote of weight $k_i$. Use Partition for reduction. 
 
   \item{\textbf{Case 2: $\alpha_1 > 2\alpha_2$.}} Construct the following instance of CWCM with $aLpLb$ as the linear ordering and with $S$ having a vote of weight $(2\alpha_1 - \alpha_2)K$ voting 
for $(a, p, b)$, $(b, a, p)$ and $(b)$ each, and a vote of weight $2(2\alpha_1 - \alpha_2)K$ voting for $(p)$. Note that here $(b, a, p)$ is the maverick. In $T$, let each $k_i$ have a vote of weight 
$(\alpha_1 - 2\alpha_2)k_i$. Use Partition for reduction.  

   \item{\textbf{Case 3: $\alpha_1 = 2\alpha_2$.}} Construct the following instance of CWCM with $aLpLb$ as the linear ordering and with $S$ having a vote of weight $K$ voting for $(b, a, p)$ and 
$(b)$ each, and a vote of weight $2K$ voting for $(a)$. Note that here $(b, a, p)$ is the maverick. In $T$, let each $k_i$ have a vote of weight $k_i$. Use Partition for reduction.  
\qed 
\end{proof}


\begin{restatable}{theorem}{nspSCav}
 \label{NSP-SC-AV}
 In 1-maverick-SP societies, for any 3-candidate scoring rule $X$ that is not isomorphic to plurality, CWCM with top-truncated votes in $X_{av}$ is NP-complete.
\end{restatable}

\begin{proof}
  Proceed the same way as in Theorem \ref{NSP-SC-AV}. Construct the following instance of CWCM with $aLpLb$ as the linear ordering and with $S$ having a vote of weight $2K$ voting for $(a, b, p)$, 
$(b, p, a)$, $(a)$, and $(b)$ each. Note that here $(a, b, p)$ is the maverick. In $T$, let each $k_i$ have a vote of weight $2k_i$. Use Fixed-difference Subset Sum for reduction. 
\qed  
\end{proof}

Finally, it is also easy to see that the following theorem holds based on the algorithm given in Theorem \ref{SP-SC-U}.

\begin{restatable}{theorem}{nspSCu}
 \label{NSP-SC-U}
 In $k$-maverick-SP societies, computing if a coalition of manipulators can manipulate $X_{\uparrow}$, for any 3-candidate scoring rule $X$, plurality$_{\downarrow}$, veto$_{\downarrow}$, 
plurality$_{av}$ with weighted top-truncated votes takes polynomial time (for any number of candidates).
\end{restatable}

Next, we look at eliminate(veto) and we show how top-truncated voting increases the complexity of manipulation for eliminate(veto) in 1-maverick-SP electorates and that it continues to portray the 
unusual behavior noted earlier.

\begin{restatable}{theorem}{eVetoNSP}
\label{eVeto-NSP}
  In 1-maverick-SP electorates, in the unique winner model, for eliminate(veto), 
  \begin{enumerate}
   \item CWCM with \textbf{complete votes} is in $P$ when the number of candidates is bounded.
   \item CWCM with \textbf{top-truncated votes} is NP-complete for even three candidates.
  \end{enumerate}
\end{restatable}

\begin{proof}
We first show how CWCM with complete votes is in $P$. To do the same, we will use the following lemma.

\begin{lemma} \label{l-eV3}
 Any successful manipulation for $p$ caused by an arbitrary set of votes of the manipulators can be achieved if all the manipulators vote for one of the single-peaked orders (that is consistent with 
the input linear ordering $L$) that has $p$ at the top. 
\end{lemma}

\begin{proof}
Suppose there was an arbitrary set of votes $W$ which resulted in a successful manipulation in favor of $p$. Let the corresponding elimination order be $e = (c_1, \cdots, c_m = p)$, where $c_i$ is 
the candidate eliminated in round $i$. Now, if the electorate were perfectly single-peaked then from Lemma \ref{l-eV1} we know that at each round only on of the rightmost or the leftmost candidate in 
the linear ordering $L$ will be eliminated. But here we have exactly 1 maverick and so that claim is not always true. As a result, we have the following cases to consider.

\item{\textbf{Case 1:}} The first case is when at each round only one of the candidates at the ends in $L$ is eliminated. This case is identical to the perfectly single-peaked case and from 
Lemma \ref{l-eV2} we know that the reverse of the elimination order here will be single-peaked. 

\item{\textbf{Case 2:}} Suppose there was some round $i$ in which there was a candidate $c_i$ eliminated such that $c_i$ was neither the leftmost nor the rightmost candidate in the $L$ remaining 
after $i-1$ rounds. If $LV(c_j)$ denotes the number of last votes for $c_j$, then the fact that $c_i$ was eliminated in the $i$th round implies that $LV(c_i) > LV(c_j), \forall j \neq i$. Now, since 
$c_i$ was not the rightmost or leftmost candidate in $L$, the only way $c_i$ can have last votes is if, in the $i$th round, the vote of the maverick had $c_i$ at the end. But then, this also implies 
that from the $(i+1)$th round, the maverick is practically a dictator as for every round from $i$ only the candidate placed last (for that round) in the maverick's vote will be eliminated. Therefore, 
in this case, the only way $p$ can win is if it is placed at the top by the maverick. As a result, what we see here is that, if at all the manipulators' votes have any significance (it wouldn't 
matter how they vote if the first candidate to be eliminated in the election was a candidate not present at the either of the ends in $L$) in determining if $p$ wins, that significance holds only for 
the first $i-1$ rounds. Also, since the first $i-1$ rounds only had one of the candidates at the ends in $L$ eliminated, we know that there is at least one single-peaked order that has $p$ at the 
top and $\{c_1, \cdots, c_{i-1}\}$ in the last $i-1$ places, where $c_j$ is the candidate placed at the $j$th position from the bottom in the preference order. 

In both the cases above, if the appropriate single-peaked order (which is the reverse of the elimination order in case 1, and a single-peaked order with $p$ at the top and $\{c_1, \cdots, c_{i-1}\}$ 
in the last $i-1$ places for case 2) is determined, then all the manipulators can simply vote for that order and this in turn will result in a successful manipulation in favor of $p$. Doing this is 
enough because we know that in eliminate(veto) any elimination order that is achieved by an arbitrary set of manipulators' votes can be achieved if all the manipulators vote the same way (by using 
the reverse of the elimination order) \cite[Lemma 12]{coleman}.\qed
\end{proof}

\textbf{Proof of Theorem \ref{eVeto-NSP} contd.} As a consequence of Lemma \ref{l-eV3}, the manipulators can now simply try out all possible single-peaked orders (that are consistent with $L$) 
which have $p$ at the top. Since there are only a polynomial number of such orders (when the number of candidates are bounded), this can be done in polynomial time. 

To prove the second part of the theorem, we can proceed the same way as in the second part of Theorem \ref{eVeto-SP}. Construct the following instance of CWCM with $aLbLp$ as the linear ordering and 
with $S$ having a vote of weight $K + 2$ voting $(b, p, a)$, a vote of weight $K - 1$ voting $(b, a, p)$, a vote of weight $1$ voting $(a, b, p)$, a vote of weight 1 voting $(p)$, a vote of weight 1 
voting $(p, a, b)$, and a vote of weight 2 voting $(a)$. Note that, here, $(p, a, b)$ is the maverick. In $T$, let each $k_i$ have a vote of weight $k_i$. Use Partition for reduction. \qed
\end{proof}

\section{Bribery} \label{bribery}
\citeauthor{faliszewski2009} \cite{faliszewski2009} were the first to look at the complexity of bribery in elections. Subsequently, the problem was studied by \citeauthor{brandt2010} 
\cite{brandt2010} in single-peaked settings and there they showed that many of the combinatorial protections for bribery vanish when the preferences are restricted to being single-peaked. 
Finally, \citeauthor{faliszewski2014} \cite{faliszewski2014} also studied the problem when the preferences are nearly single-peaked. Here, we revisit the problem of bribery in single-peaked and 
nearly-single peaked settings and we try and see if bribery too, like manipulation, fits into our theme of reinstating combinatorial protections in single-peaked and nearly single-peaked elections 
through top-truncated voting.

\subsection{Weighted-Bribery in Scoring Rules}
Here we first derive the results for 3-candidate scoring rules in single-peaked settings when only complete votes are allowed. Subsequently, we do the same when top-truncated ballots are allowed. 
The NP-completeness proofs here use an idea that is similar to the one used by \citeauthor{faliszewski2009} in Theorem 4.9 \cite{faliszewski2009}, where they use a reduction from a modified version 
of the weighted manipulation problem to show that $\alpha$-weighted-bribery is NP-complete when it isn't the case that $\alpha_2 = \alpha_3 = \cdots = \alpha_m$.  Although it is possible to extend the 
results of the complete votes case here to any $m$-candidate scoring rule, thanks to the complete characterization of weighted manipulation for scoring rules given by \citeauthor{brandt2010} 
\cite{brandt2010}, we stick to the 3-candidate case here because, to the best our knowledge, such a complete characterization is not yet known when top-truncated ballots are allowed.   

Let us first define the modified version of manipulation that we will use to reduce to the problem of weighted-bribery. The modified problem defined here is similar to the one used by 
\citeauthor{faliszewski2009} \cite{faliszewski2009}, with the only difference that in their problem all the manipulators need to have weights at least twice as much as the weight of the heaviest 
non-manipulator, while in our case we require that all the manipulators need to have weights at least thrice as much as the weight of the heaviest non-manipulator.

\begin{definition}[CWCM']
CWCM' is the same problem as CWCM with the restriction that each manipulative voter has a weight at least thrice as much as the weight of the heaviest non-manipulator. 
\end{definition}

Next, we show that for all 3-candidate scoring rules with $(\alpha_1 - \alpha_3) > 2(\alpha_2 - \alpha_3)$, CWCM' is NP-complete. The proof here makes use of the corresponding result for CWCM given 
by \citeauthor{faliszewski} in Theorem 4.4 \cite{faliszewski}.



\begin{restatable}{theorem}{spCWCM}
\label{SP-CWCM'}
 In single-peaked electorates, CWCM' with complete votes is NP-complete for 3-candidate scoring rules when $(\alpha_1 - \alpha_3) > 2(\alpha_2 - \alpha_3)$. 
\end{restatable}

\begin{proof}
 To prove the above theorem for the complete votes case, we make use of Faliszewski et al.'s proof of Theorem 4.4 \cite{faliszewski} which shows how CWCM is NP-complete for 3-candidate scoring rules 
when $(\alpha_1 - \alpha_3) > 2(\alpha_2 - \alpha_3)$. The first observation to make here is that the reduction in \cite[Theorem 4.4]{faliszewski} works even if we use Partition'. So, to prove the 
NP-completeness of CWCM', the only thing we need to do is to make sure that the manipulators used in the reduction have weights at least thrice as that of the heaviest non-manipulator. This in turn 
can be ensured if we can somehow replace every non-manipulator who has a weight greater than the threshold (which here is one-third the weight of the lightest manipulator) with several 
non-manipulators each of whose weight is less than the threshold. 

In the proof of \cite[Theorem 4.4]{faliszewski}, if we use Partition' for reduction then there will be $2n$ manipulators (instead of just 2 in the original proof) each of weight $3(2\alpha_1 - 
\alpha_2)K$ such that $n$ of them vote for $(a, p, b)$ and the other $n$ vote for $(b, p, a)$. Additionally, each of manipulators will now have a weight of $3(\alpha_1 - 2\alpha_2)a_i$, where $\{a_1, 
\cdots, a_n\}$ is an arbitrary instance of Partition' with $a_i \geq K, \forall i$ and $\sum_i a_i = 2nK$ (the rationale behind multiplying an additional factor of 3 to the weights of each 
manipulator and non-manipulator will be evident below). But now, since we need the weights of the manipulators to be at least thrice as much as the weight of the heaviest non-manipulator, we will 
have to replace each non-manipulator by at most 
\begin{equation*}
\left \lceil \frac{3(2\alpha_1 - \alpha_2)K}{\left\lfloor\frac{1}{3} (3(\alpha_1 - 2\alpha_2)) a_{\min}\right\rfloor} \right \rceil
\end{equation*}
non-manipulators, where $a_{\min} = \argmin_i\{a_i\}$. Since $a_{\min} \geq K$, and $(\alpha_1 - 2\alpha_2) \geq 1$, we have, 
\begin{equation*}
\left \lceil \frac{3(2\alpha_1 - \alpha_2)K}{\left\lfloor\frac{1}{3} (3(\alpha_1 - 2\alpha_2)) a_{\min}\right\rfloor} \right \rceil  \leq \left \lceil \frac{3(2\alpha_1 - \alpha_2)K}{a_{\min}} \right 
\rceil \leq \left \lceil C \right \rceil 
\end{equation*}  
where $C = 3(2\alpha_1 - \alpha_2)$ is a constant. Now, since this splitting can be done in polynomial time, that completes the proof of the theorem. 
\qed
\end{proof}

We now show the result for weighted-bribery in scoring rules.


\begin{restatable}{theorem}{wbSCsp}
\label{wB-SC-SP}
 In single-peaked settings, weighted-bribery with complete votes is in $P$ for 3-candidate scoring rules when $(\alpha_1 - \alpha_3) \leq 2(\alpha_2 - \alpha_3)$ and is NP-complete otherwise.
\end{restatable}

\begin{proof}
  Let $a, b,$ and $p$ be the three candidates. Without loss of generality we can assume that $\alpha_3 = 0$. We will discuss both cases ($\alpha \leq 2\alpha_2$ and $\alpha_1 > 2\alpha_2$) 
separately.

\item{\textbf{Case 1: $\alpha \leq 2\alpha_2$.}} Here we basically have two cases: i) when the input linear ordering of the candidates is $aLbLp$ and ii)  when the input linear ordering of the 
candidates is $aLpLb$. In scoring rules, it is reasonable to assume that the bribed voters will always be made to vote in such a way that $p$ will be at the top. Therefore, in the first case, since 
there is only one vote $(p, b, a)$ that has $p$ placed at the top, all the bribed voters will be bribed to vote for $(p, b, a)$. So now the only part remaining part here is to identify who to bribe. 
This can be done in the following way. Let us suppose that there was a successful bribery for a given instance. Since the only allowed votes for $aLbLp$ are $\{(a, b, p), (b, a, p), (b, p, a), (p, b, 
a)\}$, any successful bribery would have resulted in bribing $x_1$ voters initially voting $(a, b, p)$, $x_2$ voters initially voting $(b, a, p)$, and $x_3$ voters initially voting $(b, p, a)$. Now, 
to identify these $x_1, x_2,$ and $x_3$, we simply need to iterate through all the feasible (a solution $(x_1, x_2, x_3)$ is feasible if the input set of votes has at least $x_1$ voters voting for 
$(a, b, p)$, at least $x_2$ voters voting for $(b, a, p)$, and at least $x_3$ voters voting $(b, p, a)$) integral solutions of $x_1 + x_2 + x_3 \leq k$, where $k$ is the bribe limit, pick the $x_1$ 
heaviest voters voting $(a, b, p)$, the $x_2$ heaviest voters voting $(b, a, p)$, the $x_3$ heaviest voters voting $(b, p, a)$, bribe all of them vote for $(p, b, a)$, and check if it makes $p$ a 
winner. If yes, we accept. Else, we continue on to the next set of solutions which satisfy the equation. Since there are only $O(k^3)$ such solutions, and since $k \leq n$, this can be done in 
polynomial time.

For the case $aLpLb$, the only allowed votes are $\{(a, p, b), (p, a, b), (p, b, a),  \allowbreak (b, p, a)\}$, and so we can bribe a voter to vote for either $(p, a, b)$ or $(p, b, a)$. Therefore, 
the main task here is to first identify the right voters to bribe and then decide on whether to make them vote $(p, a, b)$ or $(p, b, a)$. To do this, let us assume there was a successful bribery in 
favor of $p$. This in turn would have resulted in bribing $x_1$ voters initially voting $(a, p, b)$, $x_2$ voters initially voting $(b, p, a)$, $x_3$ voters initially voting $(p, a, b)$, and $x_4$ 
voters initially voting $(p, b, a)$. Now, let $F = \{(x_1, x_2, x_3, x_4) \: | \: x_1 + x_2 + x_3 + x_4 \leq k \: \text{and $(x_1, x_2, x_3, x_4)$ is feasible} \}$. For each $x = (x_1, x_2, x_3, x_4) 
\in F$, pick the heaviest $x_1$ voters voting $(a, p, b)$, the heaviest $x_2$ voters voting $(b, p ,a)$, the heaviest $x_3$ voters voting $(p, a, b)$, and the heaviest $x_4$ voters voting $(p, b,a)$. 
Let $S'$ be the set of all these voters. Next, calculate the scores of $p, a,$ and $b$ considering all the votes in $S-S'$, where $S$ is the set of all votes given as part of the input. Let these 
scores be denoted by $s_{S-S'}(p), s_{S-S'}(a),$ and $s_{S-S'}(b)$ respectively. Now for the case $aLpLb$, since $\alpha \leq 2\alpha_2$, for any set of votes, $p$ cannot lose to both $a$ and $b$. 
Therefore the only possible cases for the scores $s_{S-S'}(p), s_{S-S'}(a),$ and $s_{S-S'}(b)$ are the following. For each of them we mention how the bribery needs to be done.

\begin{enumerate}
 \item $ s_{S-S'}(p) \geq s_{S-S'}(a), \: s_{S-S'}(p) \geq s_{S-S'}(b) $: In this case, for each voter in $S'$, make him or her vote either $(p, a, b)$ or $(p, b, a)$, if they aren't already voting 
for either of  them.
 \item $s_{S-S'}(b) > s_{S-S'}(p), \: s_{S-S'}(a) \leq s_{S-S'}(p)$: In this case, for each voter in $S'$, make him or her vote $(p, a, b)$, if they aren't already. 
 \item $s_{S-S'}(a) > s_{S-S'}(p), \: s_{S-S'}(b) \leq s_{S-S'}(p)$: Here, for each voter in $S'$, make him or her vote $(p, b, a)$, if they aren't already. 
\end{enumerate}

Doing the above takes only polynomial time since there are only $O(k^4)$ feasible solutions and $k \leq n$. 

\item{\textbf{Case 2: $\alpha > 2\alpha_2$.}} To show that weighted-bribery in this case is NP-complete, we use a reduction from CWCM'. Note that to show this we use the linear ordering $aLpLb$ 
since  there are polynomial time algorithms for both CWCM' and weighted-bribery when $p$ is at one of the ends in the societal order. Now, given an arbitrary instance $(C, S, T, p)$ of CWCM', we 
construct the following instance $(C, V, p, k)$ where $C$ is set of candidates, $V = S \cup T'$ where $T'$ is the set of voters from $T$ voting for either $(a, p, b)$ or $(b, p, a)$. We set $k = 
|T|$.

It is easy to see that if a successful manipulation exists then bribery with at most $k = |T|$ bribes is possible. To prove the converse we basically need to show that any bribery with at most $|T|$ 
bribes is possible if we bribe only the voters from $T'$. For this, let us assume this were not the case and that there was a voter $v \in V - T'$ who was part of the successful bribery in favor of 
$p$. By bribing $v$, the gain for $p$ against $a$ or $b$ is at most $(\alpha_1 + (\alpha_1 - \alpha_2))w(v)$, where $w(v)$ is the weight of voter $v$. On the other hand if a voter $v' \in T'$ was 
bribed, then the gain for $p$ against $a$ or $b$ is at least $(\alpha_1 - \alpha_2)w(v')$. But then, according to the restriction imposed on CWCM', any $v' \in T'$ has a weight at least thrice as 
much 
as the heaviest voter in $V' - T$. This in turn implies that $w(v') \geq 3w(v)$, and so we might as well bribe $v'$ instead of $v$. As a result, any bribery in favor of $p$ can be achieved by bribing 
only the voters from $T'$. Therefore, existence of a successful bribery implies that a successful manipulation exists in CWCM'. \qed
\end{proof}

For the case of top-truncated ballots, we proved in Theorem \ref{SP-SC-D} that, in single-peaked settings, CWCM with top-truncated votes is NP-complete for all 3-candidate scoring rules except 
plurality and veto when the evaluation scheme was round-down. Similarly, we also showed in Theorem \ref{SP-SC-AV} that CWCM with top-truncated votes is NP-complete for all 3-candidate scoring rules 
except plurality when the evaluation scheme was average-score. Based on these two theorems it is easy to see that we can make similar `splitting' arguments as in Theorem \ref{SP-CWCM'} to prove these 
results hold true even for CWCM'. As a result, we can state the following results which can be proved by using a reduction from CWCM' similar to the one shown in case 2 of Theorem \ref{wB-SC-SP}.


\begin{theorem}
 For any 3-candidate scoring rule $X$ that is not isomorphic to plurality or veto, in single-peaked electorates, weighted-bribery with top-truncated votes in $X_{\downarrow}$ is NP-complete.
\end{theorem}

\begin{theorem} 
 For any 3-candidate scoring rule $X$ that is not isomorphic to plurality, in single-peaked electorates, weighted-bribery with top-truncated votes in $X_{av}$ is NP-complete.
\end{theorem}

Note that we can prove the corresponding results for nearly single-peaked electorates as well.

\subsection{Weighted-Bribery in Eliminate(veto)}
Here we look at the problem of weighted-bribery in eliminate(veto). First, we study the problem in single-peaked electorates and following that we look at the nearly single-peaked case. In both 
cases, yet again, we observe that allowing top-truncated voting increases the complexity of weighted-bribery from being in $P$ to being NP-complete. 

\begin{restatable}{theorem}{wbEVsp}
\label{wB-eV-SP}
In single-peaked electorates, in the unique winner model, for 3-candidate eliminate(veto), 
  \begin{enumerate}
   \item weighted-bribery with \textbf{complete votes} is in $P$.
   \item weighted-bribery with \textbf{top-truncated votes} is NP-complete.
  \end{enumerate}
\end{restatable}

\begin{proof}
 First we discuss the polynomial time algorithm for weighted-bribery when only complete votes are allowed.
 
 Let $a, b,$ and $p$ be the three candidates. Without loss of generality, we need to consider only two linear orders: $aLbLp$ and $aLpLb$.

\item{\textbf{i) $aLbLp$:}} For the given linear ordering, only the following elimination orders are possible: $\{(a, b, p), (a, p, b), (p, a, b), (p, b, a)\}$, where $(a, b, p)$ implies that 
$a$ is eliminated in the first round, followed by $b$, and then $p$ (which is the winner). We will consider each of these cases separately.

\item{\textbf{Case 1:}} $(a, b, p)$. Here no bribery is required as $p$ is already the winner.

\item{\textbf{Case 2:}} $(a, p, b)$. In this case it is easy to see that a simple greedy strategy would suffice. Since $b$ cannot be eliminated in the first round (as there is no vote with $b$ at the 
end according to $L$), we only need to make sure that $p$ wins against $b$ in the second round. This can be ensured by doing the following: 

Starting with the heaviest voter among all voters who do not rank $p$ above $b$, bribe him or her to vote for $(p, b, a)$. If $p$ is the winner after the bribe, then accept. If not, 
continue the same until number of voters bribed exceeds the limit $k$. If the limit exceeds, reject. 

\item{\textbf{Case 3:}} $(p, a, b)$. Here for $p$ to win, it has to be saved from elimination in the first round, and subsequently has to beat $b$ in the second round. To do this, we can proceed in 
two phases: In the first phase, starting with the heaviest voter among all voters who do not rank $a$ last, bribe them to vote for $(p, b, a)$. After each bribe, check if $LV(a) > LV(p)$, where 
$LV(a)$ denotes the total weight of all the voters who have placed $a$ last in their preference ordering. Once that is done, i.e. once $LV(a) > LV(p)$, we check if $p$ is the winner. If yes, we 
accept. Else, we move on to the second phase where starting from the heaviest voter but now among all voters who do not rank $p$ above $b$, we bribe them to vote for $(p, b, a)$. Here again, after 
each bribe, we see if it results in $p$ being the winner. If it does, we accept. Otherwise, if at any point during this algorithm the number of voters bribed exceeds the limit $k$, then we reject.   

\item{\textbf{Case 4:}} $(p, b, a)$. Follow the exact same algorithm as in case 3. 

\item{\textbf{ii) $aLpLb$:}} Here, only the following elimination orders are possible: $\{(a, b, p), \allowbreak (a, p, b), \allowbreak (b, a, p), \allowbreak (b, p, a)\}$. Again, we will consider 
each of these cases separately. 

\item{\textbf{Case 1:}} $(a, b, p)$, $(b, a, p)$. In both these cases no bribery is required as $p$ is already the winner.

\item{\textbf{Case 2:}} $(a, p, b)$. Here the situation is a little more complicated since the voters can be bribed to either induce the elimination order $(a, b, p)$ or $(b, a, p)$. But although 
this is the case, it is possible to prove (see Lemma \ref{wB-eV-SP-L}) that inducing $(a, b, p)$ is enough as inducing the other would require bribing more voters. So now, this case is similar to 
case 2 for the linear order $aLbLp$ and we can use the use the algorithm described there. 

\item{\textbf{Case 3:}} $(b, p, a)$. This is similar to case 2 above with $a$ and $b$ interchanged. 

This concludes the proof of the first part.

To prove the second part, we show how Partition' can be used to prove that weighted-bribery is hard for eliminate(veto) when top-truncated votes are allowed. The proof essentially follows an idea 
similar to the one used for proving the corresponding result for manipulation in Theorem \ref{eVeto-SP}, only that we show this reduction through the modified partition problem. The main motivation 
behind using the modified partition problem is to basically facilitate a situation that will enable us to argue that the bribery can be restricted to a certain set of voters (who correspond to the 
manipulators in the corresponding manipulation instance). 

Now, it is easy to see that the problem is clearly in NP. To prove NP-hardness, we show a reduction from an arbitrary instance $\{a_1, \cdots, a_n\}$ of Partition', where $\sum_i a_i = 2nK$ and $a_i 
\geq K, \forall i$, to an instance of weighted-bribery $(C, V, p, k)$, where $C = \{a, b, p\}$ is the set of candidates, $aLbLp$ is the linear order of the candidates, and $V$ is the set of the 
following voters. 

\begin{enumerate}
 \item For each $a_i$, construct a voter $v_i$ whose weight is $a_i$ and who votes for $(b, a, p)$. Let $T$ be the set of all these voters.
 \item Construct the following set of voters $S$: $n$ voters of weight $K$ each voting for $(b, p, a)$, $n$ voters of weight $K-1$ each voting for $(b, a, p)$, $n-1$ voters of weight $1$ each voting 
for $(b, a, p)$, 1 voter of weight 1 voting for $(a, b, p)$, 1 voter of weight 2 voting for $(b, p, a)$, 1 voter of weight 2 voting for $(p)$, and 1 voter of weight 3 voting for $(a)$. 
\end{enumerate}
Set the bribe limit $k = n$ and $V = S \cup T$. 

First thing to notice is that, if $a$ is eliminated in the first round then $p$ cannot be a unique winner, no matter how we choose the $n$ voters to bribe. Therefore the only way to make $p$ a 
unique winner is by making sure that $a$ does not get eliminated. Now since the bribe limit is $n$, the maximum total weight of the voters bribed cannot exceed $2nK$. Also, using arguments similar 
to the one used in the proof of Theorem \ref{eVeto-SP}, we can see that any bribery with at most $n$ bribes is possible only if there is an equal weight of $nK$ that is assigned to both $(p)$ and 
$(a)$. Therefore, the only part that remains to be argued is to show that this bribery with at most $n$ bribes can be achieved if we bribe only those voters in $T$. 

Before we show this, let us see the objectives involved in the bribery of this instance. First is that any bribe results in changing the vote of the corresponding voter to either $(p)$ or $(a)$. 
Secondly, while bribing a voter, we always want the gain for $p$ against every other candidate to be as much as possible. And lastly, we also need to ensure that while $a$ is given enough score so 
that it is not eliminated in the first round, any bribe to vote for $(a)$ does not result in the score of $p$ decreasing.

Now, having seen the objectives, let us suppose that there was a voter $v \in S$ who was part of a successful bribery in favor of $p$. Before the bribe, $v$ has one among $(b, p, a)$, $(b, a, p)$, 
$(a, b, p)$, $(p)$, and $(a)$ has his or her vote. But given the objectives, the only rational choice of a voter who has to be bribed is one who's current vote is $(b, a, p)$ (bribing one with $(a)$ 
is irrational since to compensate the reduction in the score of $a$ we will have to bribe another voter; bribing a voter with $(b, p, a)$ is irrational since if this voter was bribed to vote for 
$(p)$ then the gain for $p$ would have been more if instead a voter voting $(b, a, p)$ would have been bribed, or otherwise if he or she was bribed to vote for $(a)$ then another voter would have to 
be bribed to compensate the reduction in $p$'s score; bribing a voter with $(a, b, p)$ is irrational since bribing this voter has the same effect as bribing a voter with $(b, a, p)$ and the weights 
of all the voters voting $(b, a, p)$ is at least as much as the weight of the voter voting $(a, b, p)$). But then, all the voters in $T$ are not only voting for $(b, a, p)$ initially, but they also 
have weights which are at least as much as of those voters voting $(b, a, p)$ in $S$. Hence, any bribery that is possible with at most $n$ bribes can be achieved if we bribe only the voters from $T$. 
Now, since the total weight of all voters in $T$ is equal to $2nK$, this implies that $p$ is a unique winner if and only if there is a partition. 
\qed 
\end{proof}

\begin{restatable}{lemma}{wbEVl} 
\label{wB-eV-SP-L}
 For the linear order $aLpLb$, when the input elimination order is $(a, p, b)$, for a successful bribery it is enough to induce the elimination order $(a, b, p)$.
\end{restatable}

\begin{proof}
 Since the linear order is $aLpLb$, only the following votes are allowed: $\{(a, p, b), \allowbreak (p, a, b), \allowbreak (p, b, a), (b, p, a)\}$. Let us assume that the total weight of all the 
voters voting $(a, p, b)$ is $W_1$, those voting $(p, a, b)$ is $W_2$, those voting $(p, b, a)$ is $W_3$, and those voting $(b, p, a)$ is $W_4$. Now, since the input elimination order is $(a, p, b)$, 
we have the following inequalities: (1) $W_1 + W_2 < W_3 + W_4$ (since $a$ is getting eliminated in the first round, $LV(a) > LV(b)$, where LV(a) is the total weight of all the voters who place $a$ at 
the end), and (2) $W_1 + W_2 + W_3 < W_4$ (since $p$ gets eliminated in the second round). A successful bribery for $p$ can proceed in two ways. First is if the bribery induces the elimination order 
$(a, b, p)$ and the other way is to induce the elimination order $(b, a, p)$. The task here is to prove that inducing $(a, b, p)$ is better than inducing $(b, a, p)$, or in other words, prove that 
inducing $(b, a, p)$ requires bribing at least one more voter than what is required to induce $(a, b, p)$. 

To prove this, let us consider the case when the elimination order $(a, b, p)$ is induced. Since the input already has $a$ getting eliminated in the first round, we only need to look at the second 
round and see how much total weight should be bribed so that $p$ is saved from elimination. Once $a$ is eliminated, $LV(b) = W_1 + W_2 + W_3$, and $LV(p) = W_4$. Therefore, for $p$ to win in the 
second round, the weight of votes we need to bribe is at least $D_{pb} = \left \lceil \frac{W_4 - (W_1 + W_2 + W_3)}{2} \right \rceil$. Let this bribery require $k_1$ voters to bribed. Moreover, note 
that all the bribed votes are of the form $(b, p, a)$ and they will be bribed to vote for $(p, b, a)$.  

Now, let us consider the second elimination order $(b, a, p)$. To induce this, that is to make $b$ get eliminated in the first round, the total weight of votes we need to bribe is at least $D_{ab} = 
\left \lceil \frac{(W_3 + W_4) - (W_1 + W_2)}{2} \right \rceil = \left \lceil \frac{2W_3 + W_4 - (W_1 + W_2 + W_3)}{2} \right \rceil = W_3 + \left \lceil \frac{ W_4 - (W_1 + W_2 + W_3)}{2} \right 
\rceil = W_3 + D_{pb}$. As a result, among voters voting $(b, p, a)$ we need to bribe a weight of at least $D_{pb}$ so as to get $b$ eliminated in the first round. But then, this bribery requires at 
least $k_1$ voters (from above), which in turn implies that we need at least one more voter to be bribed here than in the case where the elimination order $(a, b, p)$ was induced. Hence, when the 
input elimination order is $(a, p, b)$ and the linear order $aLpLb$, for a successful bribery it is enough to induce the elimination order $(a, b, p)$.  \qed
\end{proof}

%
Next we look at the complexity of bribery for eliminate(veto) in 1-maverick single-peaked electorates.

\begin{restatable}{theorem}{wbEVnsp}
\label{wB-eV-NSP}
In 1-maverick single-peaked electorates, in the unique winner model, for 3-candidate eliminate(veto), 
  \begin{enumerate}
   \item weighted-bribery with \textbf{complete votes} is in $P$.
   \item weighted-bribery with \textbf{top-truncated votes} is NP-complete.
  \end{enumerate}
\end{restatable}

\begin{proof}
 To prove the first part, we proceed like in Theorem \ref{wB-eV-SP} and consider the two possible input linear orderings $aLbLp$ and $aLpLb$ separately, where $a, b,$ and $p$ are the three candidates.
 
\item{\textbf{i) $aLbLp$:}} Now as opposed to the case in single-peaked electorates, here  $b$ can be eliminated in the first round (due to the presence of the maverick) and so, the 
following elimination orders are possible: $\{(a, b, p), (b, a, p), (b, p, a), (a, p, b), \allowbreak (p, a, b), (p, b, a)\}$. We will consider each of these separately.

\item{\textbf{Case 1:}} $(a, b, p), (b, a, p).$ In both these cases, no bribery is required as $p$ is already the winner.
\item{\textbf{Case 2:}} $(b, p, a).$ $b$ can only be  placed in the last place in his or her preference ordering by the maverick. Therefore, here, the fact that $b$ is getting eliminated in the first 
round implies that the weight of the maverick is greater than the total weight of all voters who have placed $a$ at the end and the total weight of all voters who have placed $p$ at the end. This in 
turn implies that once $b$ is eliminated, only the candidate above $b$ in the maverick's preference order will get eliminated, or in other words the maverick in this scenario acts like a dictator. 
Hence, the only voter who needs to be bribed is the maverick. 
\item{\textbf{Case 3:}} $(a, p, b), (p, a, b), (p, b, a).$ For all these cases, we can follow the exact same algorithm as outlined in the proof of Theorem \ref{wB-eV-SP} for the case of 
single-peaked electorates.

While we have covered all the cases for the linear order $aLbLp$, one question that could arise is: What if the input set of voters did not have a maverick? Is there anything more that can be 
achieved if we introduce a maverick while bribing? The answer to this is No. This is so because, consider the case if the input elimination order was $(a, p, b)$. Then the only rationale choice 
(among $(a, p, b)$ and $(p, a, b)$) of vote for a bribery-induced maverick according to $L$ is $(p, a, b)$. Now, let us suppose that we were to introduce such a maverick through bribing. This in turn 
can result in two possibilities: i) $b$ getting eliminated in the first round ii) $a$ remains the one to be eliminated in the first round. In case i), the fact that $b$ got eliminated in the first 
round as a result of introducing the maverick implies that there already existed a voter in the input whose weight was so high that bribing him alone would have anyway sufficed. And our greedy 
strategy would have anyway identified such a voter and bribed him or her to vote for $(p, b, a)$. In the other case, the only objective this achieves is that it lowers the gap between the scores of 
$p$ and $b$, which again would have been achieved by our greedy strategy as it always picks the heaviest voter who does not rank $p$ above $b$ and makes him or her vote for $(p, b, a)$. Therefore 
there is nothing more we are achieving by introducing a maverick through bribery. Similarly, we can argue the same for the other elimination orders.

\item{\textbf{ii) $aLpLb$:}} For this linear order if the input elimination order were $(p, a, b)$ or $(p, b, a)$ then we can make the same argument as in the case 2 above to say that we need to 
bribe only the maverick and make him or her vote for $(p, a, b)$ or $(p, b, a)$. For all the other cases we can proceed exactly as outlined in the proof of Theorem \ref{wB-eV-SP} for the case of 
single-peaked electorates. Additionally, the question of introducing a maverick through bribing does not arise here because both the maverick votes $(a, b, p)$ and $(b, a, p)$ are irrational choices 
as they both have $p$ placed at the end. 

This concludes the proof of the first part of theorem. 

Next, we show the NP-completeness result by using similar arguments as in the proof of second part of Theorem \ref{wB-eV-SP} with only a slight modification to the construction of $T$. Here, in $T$ 
we have the following voters: $n$ voters of weight $K$ each voting for $(b, p, a)$, $n$ voters of weight $K-1$ each voting for $(b, a, p)$, $n-1$ voters of weight 1 each voting $(b, a, p)$, 1 voter 
of weight 1 voting for $(a, b, p)$, 1 voter of weight 1 voting for $(p)$, 1 voter of weight 1 voting for $(p, a, b)$, 1 voter of weight 2 voting for $(b, a, p)$, and 1 voter of weight 2 voting for 
$(a)$. Note that the voter voting $(p, a, b)$ is the maverick. \qed
\end{proof}
  
\subsection{Is Weighted-Bribery for Weak-Condorcet Consistent Rules Always Easy?}



\citeauthor{brandt2010} \cite{brandt2010} showed that in single-peaked electorates weighted-bribery is in $P$ for all weak-Condorcet consistent voting rules (see \cite[Theorem 4.4]{brandt2010} 
for a more general result). The $P$ results in their theorem and the reason why it was possible to consider all weak-Condorcet consistent voting rules together was because of the well-known 
property of single-peaked electorates where it is guaranteed that there is always at least one weak-Condorcet winner (the top choices of the ``median'' voters are always weak-Condorcet winners). 
However, this property no longer holds when top-truncated votes are allowed. As has also been pointed out by \citeauthor{cantala2004} \cite{cantala2004}, it is not even necessary that a 
weak-Condorcet winner exists in such settings. We illustrate this with the following example.

\begin{example}
 Let $C = \{a, b, c, d, e\}$ with $a L b L c L d L e$ as the linear ordering. Let there be 5 voters and let their votes be $a\succ_{v_1} b \succ_{v_1} c$, $b \succ_{v_2} c$, $c \succ_{v_3} d$, $d 
\succ_{v_4} e$, and $e \succ_{v_5} d$, respectively. Now, it is easy to see that in the pairwise majority relation, $a$ and $b$ lose to $d$, $c$ loses to $b$, and both $d$ and $e$ lose to $c$. Since 
everyone loses at least once, there is no weak-Condorcet winner.  
\end{example}

Because of the above we can no longer consider all weak-Condorcet consistent rules together like in \cite{brandt2010} and exploit the connection between the weak-Condorcet winner(s) and ``median'' 
voters to come up with polynomial time algorithms for weighted-bribery. In fact, next we show that for 3-candidate Baldwin's rule (also known as Fishburn's version of Nanson's rule \cite{niou1987}), 
which is a weak-Condorcet consistent rule in single-peaked electorates \cite{brandt2010}, weighted-bribery is NP-complete when top-truncated ballots are allowed. To show this we use an idea similar 
to the one used in Theorem \ref{wB-eV-SP}. 

\begin{restatable}{theorem}{wbBaldwin}
\label{wB-Baldwin-SP-PV}
 In single-peaked electorates, weighted-bribery with \textbf{top-truncated votes} is NP-complete for 3-candidate Baldwin's rule.
\end{restatable}

\begin{proof}
 The problem is clearly in NP. To prove NP-hardness, we show a reduction from an arbitrary instance $\{a_1, \cdots, a_n\}$ of Partition', where $\sum_i a_i = 2nK$ and $a_i \geq K, \forall i$, 
to an instance of weighted-bribery $(C, V, p, k)$, where $C = \{a, b, p\}$ is the set of candidates, $aLpLb$ is the linear order of the candidates, and $V$ is the set of the following voters. 
\begin{enumerate}
 \item For each $a_i$, construct a voter $v_i$ whose weight is $2a_i$ and who votes for $(b)$. Let $T$ be the set of all these voters.
 \item Construct the following set of voters $S$: $n$ voters of weight $2K$ each voting for $(b, p, a)$, $n$ voters of weight $2K$ each voting for $(b)$, $n$ voters of weight $K$ each voting for 
$(a)$, 2 voters of weight 1 each voting for $(b, p, a)$, 2 voters of weight 1 each voting for $(a)$, and 1 voter of weight 1 voting for $(p)$. 
\end{enumerate}
Set the bribe limit $k = n$ and $V = S \cup T$. 

First, observe that if $a$ is eliminated in the first round then $p$ cannot win, no matter how we choose the $n$ voters to bribe. Therefore the only way to make $p$ a winner is by making sure that 
$a$ does not get eliminated. Now since the bribe limit is $n$, the maximum total weight of the voters bribed cannot exceed $4nK$. Also note that we need to bribe at least $4nK$ weight because 
otherwise either $p$ or $a$ will be eliminated in the first round. Additionally, it is possible to show that even now $p$ can be a winner only if there is a total weight of $2nK$ voting for $(p, a, 
b)$ and a total weight of $2nK$ voting for $(a, p, b)$ (informally, this argument holds because failing to have equal votes on $(p, a, b)$ and $(a, p, b)$ will cause either $p$ or $a$ to be eliminated 
in the first round). Therefore, the only part that remains to be argued is to show that any bribery with at most $n$ bribes can be achieved if we only bribe the voters in $T$.  

To prove this, let us assume the contrary and suppose that there was a voter $v \in S$ who was part of a successful bribery in favor of $p$. Before the bribe, $v$ has one among $(b, p, a)$, 
$(b)$, $(p)$ and $(a)$ has his or her vote. Now, given the fact that we always want the gain for $p$ to be as much as possible, and because $a$ needs to be given enough score to prevent it from being 
eliminated first, the only rational choice of voter who has to be bribed is one who's current vote is $(b)$ (bribing one with $(a)$ is irrational since to compensate the reduction in the score 
of $a$ we will have to bribe another voter, bribing a voter with $(b, p, a)$ is irrational since if this voter was bribed to vote for $(p, a, b)$ or $(a, p, b)$ then the gain for $p$ would have been 
more if instead a voter voting $(b)$ was bribed). But then, all the voters in $T$ are not only voting for $(b)$ initially, but they also have weights which are at least as much as of those voters 
voting $(b)$ in $S$. Hence, any bribery that is possible with at most $n$ bribes can be achieved if we bribe only the voters from $T$. Now, since the total weight of all voters in $T$ is equal 
to $4nK$, this implies that $p$ can be a winner if and only if there is a partition in Partition'. 
\qed
\end{proof}

Note that almost the same proof as above can be used to show that in single-peaked electorates CWCM with top-truncated ballots is NP-complete for 3-candidate Baldwin's rule. 

\section{Is Allowing Top-truncated Voting in Single-Peaked Electorates Always Beneficial?} \label{avr}
Although we have seen instances like in 3-candidate scoring rules with round-up evaluation scheme where the complexity of manipulation decreases as a result of moving from a purely single-peaked 
setting to a setting where top-truncated votes are allowed, we haven't really seen examples of any other voting rule which shows this behavior. Moreover, we also know that with a different evaluation 
scheme like round-down or average-score this behavior is no longer seen for even 3-candidate scoring rules. Therefore, a natural question one could ask is: ``What role does the evaluation scheme 
play? Is it possible that given a voting rule one can always construct an evaluation scheme so that it will be beneficial to allow top-truncated voting in single-peaked electorates?''. Alternatively, 
one could also ask: ``Is there a voting system for which it is always easy to manipulate when top orders are allowed?''. We answer the former question in the negative and the latter one in the 
affirmative. We show that, as long as all the unranked candidates are assumed to be tied and are assumed to be ranked below the ranked candidates (which is the natural definition of a top-truncated 
vote), there is at least one voting system for which, irrespective of how the top-truncated votes are dealt with, it is NP-hard to manipulate in purely single-peaked settings, but is easy to 
manipulate when top-truncated votes are allowed.  

\begin{restatable}{theorem}{avr}
 There exists a voting system for which, in single-peaked settings, 
 \begin{enumerate}
  \item CWCM with \textbf{complete votes} is NP-complete.
  \item CWCM with \textbf{top-truncated votes} is in P.
 \end{enumerate}  
\end{restatable}

\begin{proof}
Let us first define the (artificial) voting system that we consider here.


\begin{definition}[Artificial Voting Rule (AVR)]
 Given the set $V$ of voter preferences, the score of each candidate $c \in C$ is calculated as
  \begin{align*}
    s(c) =  \displaystyle \sum_{v \in V} \displaystyle \sum_{\substack{a \in C \\ {a \neq c} \\ {c \succ_{v} a}}} (m - pos(c))w(v)  
  \end{align*}  
  where $pos(c)$ is the position of candidate $c$ in the preference order of $v$ if it is ranked ($pos(c) = i$ if $c$ is ranked in the $i$th position by $v$) and is $m$ otherwise, $m = |C|$, $w(v)$ 
is the weight of voter $v$, and $c \succ_{v} a$ denotes that $c$ is ranked above $a$ by $v$. 
\end{definition}

\begin{example}
 Let $C = \{a, b, c, d, e\}$ be the set of candidates and let there be a single voter of weight $w$ with the preference ordering $d \succ b \succ a \succ e \succ c$. Then the scores of the candidates 
are $s(a) = 4w$, $s(b) = 9w$, $s(c) = 0$, $s(d) = 16w$, and $s(e) = w$.    
\end{example}

So now we need to show that for AVR, in single-peaked settings, CWCM with complete votes is NP-complete while CWCM with top-truncated votes is in P. 

We show the first part by a reduction from Partition. Given an arbitrary instance $\{k_i\}_{1\leq i \leq t}$, $\sum_i k_i = 2K$, of Partition, construct the following instance of CWCM, where $a, b,$ 
and $p$ are the three candidates and $aLpLb$ is the linear ordering over the candidates. In $S$, let there be two voters of weight $7K$ each voting for $(a, p, b)$ and $(b, p, a)$ respectively, and 
two voters of weight $K$ each voting for $(p, a, b)$ and $(p, b, a)$ respectively. In $T$, let each $k_i$ have a vote of weight $k_i$. According to the rule defined above, the scores of $a, b,$ and 
$p$ are $29K, 29K,$ and $22K$ respectively.

Suppose there exists a partition. Let those manipulators in one half vote $(p, a, b)$ while those in the other half vote $(p, b, a)$. As a result the scores of $a, b,$ and $p$ are all $30K$ and hence 
$p$ is a winner. 

Conversely, suppose there exists a manipulation in favor of $p$. In AVR it is reasonable to assume that all the manipulators place $p$ first. So, now, let $x$ and $y$ be total weight of the 
manipulators in $T$ who vote $(p, a, b)$ and $(p, b, a)$ respectively. Since there exists a successful manipulation in favor of $p$, the score of $p$ should be at least as much as that of 
$a$. Therefore, we have: $22K + 4x + 4y \geq 29K + x$. Using the fact that $x + y = 2K$, this simplifies to $x \leq K$. Doing the same with respect to $p$ and $b$ we have, $y \leq K$. But 
then since $x + y = 2K$, this implies that $x = K$ and $y = K$ and that there exists a partition.     

For the second part, it is easy to see that the optimal strategy for the manipulators is to just vote $(p)$.\qed
\end{proof}

%% file: conclusion.tex
\section{Conclusion and Future Work}
The central theme of this paper was the reinstatement of combinatorial protections in single-peaked and nearly-single peaked electorates by allowing top-truncated voting. We observed this behavior 
first in the case of manipulation and showed how for different voting protocols manipulation with complete votes was in $P$ whereas manipulation with top-truncated votes jumped to being NP-complete. 
These results were followed by the results for bribery where, again, we observed similar behavior for the voting rules considered. In studying the above two, we note that, to the best of our 
knowledge, we are the first to systematically look at the impact on complexity of manipulative actions when the electorate is single-peaked or nearly single-peaked and when top-truncated preferences 
are allowed. In addition to the above results, we also showed an instance of a natural voting system (eliminate(veto)) where, contrary to intuition, the complexity of manipulation, when top-truncated 
ballots are allowed, actually increases from being in $P$ in the general case to being NP-complete in the single-peaked case. Finally, we concluded our discussion by showing the example of a voting 
system where allowing top-truncated voting isn't beneficial in the sense that it actually results in a decrease in the complexity of manipulation.


There are many possible avenues for future work. Foremost would be look at some other interesting voting rules and also consider other types of partial preferences (like bottom orders, weak orders 
etc.) to try and see if similar behavior is observed in them. Second, in this paper we have considered only manipulation and bribery, but not control. Therefore, we feel that it would be worthwhile 
to see if similar observations can be made for the problem of control as well. Third, while considering nearly single-peaked preferences in this paper, we have essentially talked about only one 
notion of nearness, namely, the $k$-maverick notion. However, there are several other notions of nearness (see \cite{erdelyi2013}) and seeing if we can obtain similar results for them as well would 
be interesting. Finally, we have considered only weighted elections in this work, but we believe that looking at the unweighted case would be very interesting and it is definitely something we 
consider as a future research direction.